\documentclass[12pt]{article}
\usepackage[english]{babel}
\usepackage[english]{babel}
\usepackage{amssymb}
\usepackage{amsmath}
\usepackage{amsfonts}
\usepackage{amsbsy}
\usepackage{indentfirst}
\usepackage{graphicx}
\usepackage{graphics}
\usepackage{epsfig}
\usepackage{color}

\newtheorem{teor}{Theorem}[section]
\newtheorem{pro}{Proposition}[section]

\newtheorem{lema}{Lemma}[section]
\newtheorem{Exa}{Example}[section]
\newenvironment{proof}[1][Proof]{\noindent\textbf{#1.} }{\ \rule{0.5em}{0.5em}}
\begin{document}
\newcommand{\bfA}{\mathbf{A}}
\newcommand{\C}{\mathbb{C}}
\newcommand{\R}{\mathbb{R}}
\newcommand{\N}{\mathbb{N}}
\newcommand{\E}{\mathbb{E}}
\newcommand{\bfgr}{\boldsymbol{\nabla}}
\newcommand{\bfal}{\boldsymbol{\alpha}}
\newcommand{\bfpi}{\boldsymbol{\pi}}
\newcommand{\bfta}{\boldsymbol{\tau}}
\newcommand{\bfr}{\mathbf{r}}
\newcommand{\bfq}{\mathbf{q}}
\newcommand{\e}{\mathrm{e}}
\def\d{{\rm d}}
\def\tr{{\rm tr}}
\def\Tr{{\rm Tr}}
\def\hS{{\hat S}}
\def\al{\alpha}
\def\1{{\bf 1}}
\def\g{\gamma}
\font\titlefont=cmbx10 scaled\magstep1

\title{
Towards non-Hermitian quantum statistical thermodynamics
}
\author{N. Bebiano\footnote{ CMUC, University of Coimbra, Department of
Mathematics, P 3001-454 Coimbra, Portugal (bebiano@mat.uc.pt)},
J.~da Provid\^encia\footnote{CFisUC, Department of Physics,
University of Coimbra, P 3004-516 Coimbra, Portugal
(providencia@teor.fis.uc.pt)}~
and J.P. da
Provid\^encia\footnote{Department of Physics, Univ. of Beira
Interior, P-6201-001 Covilh\~a, Portugal
(joaodaprovidencia@daad-alumni.de)}}
\maketitle
\begin{abstract}
Non-Hermitian Hamiltonians possessing a discrete real spectrum motivated a
remarkable research activity
in quantum physics and new insights have emerged.
In this paper we formulate concepts of statistical thermodynamics
for systems
described by non-Hermitian Hamiltonians with real eigenvalues.
We mainly focus on the case where the energy and another observable are the conserved quantities.
The notion of entropy and entropy inequalities are central in our approach, which
treats equilibrium thermodynamics.
\end{abstract}
\section{Introduction}\label{S1}
Entropy is a fundamental concept in science, with origin in thermodynamics. This branch of physics
has been founded almost 200 years ago and its conceptual bases remained unchanged until now.
The realm of thermodynamics  has been considerably extended,
from dealing with macroscopic systems, to individual quantum systems and black holes. Jaynes, in 1957, revisited the formulation
of thermodynamics for an arbitrary number of conserved quantities, through the maximum entropy principle \cite{jaynes}.
This principle demands the following.
Consider a given system and specified constraints on its conserved properties, which, necessarily, are always very far from uniquely determining
the macroscopic state of the system. This state is obtained, according with the well known Boltzmann prescription, by considering only the
microstates (in the sense of analytical mechanics) which are compatible with the constraints, and by  assigning equal probability to each one of them.
The described procedure is often simplified by replacing the exact constraints by the respective average values for a certain statistical population.
It is well known that the relative statistical error involved in this simplification is of the order of
the inverse of the square root of the number of particles of the system.
The {\it  maximum entropy principle} states that the probability distribution which better represents the
equilibrium state is the one which maximizes the entropy, under imposed average values of certain conserved quantities, i.e., constants of motion.

In the last decades, the consideration, in quantum physics,  of
non-Hermitian Hamiltonians with a real discrete spectrum,
gave rise to an intense research activity
in physics and mathematics, see e.g., \cite{bagarello*,bagarelo**,*,[1],mostafa,mostafa1,providencia,scholtz}.
In this note, the
{maximum entropy principle} is formulated for systems described by the non-Hermitian Hamiltonian $H$ with a real discrete spectrum.
We focus on the extension, for this setup, of  classical results of thermodynamics, namely, a fundamental inequality
which reflect the second law,
and related topics \cite{Bebiano Lemos Prov'04,Bebiano Lemos Prov'05,vonneumann}. As a starting point, we reinterpret some of the standard thermodynamic quantities in the non-Hermitian context.

We propose a formulation of equilibrium statistical thermodynamics in this framework.
We assume that
$H$ is defined in a Hilbert space $\cal H$ with inner product $\langle\cdot,\cdot\rangle$.
In this space,
a new inner product is introduced in order to preserve the standard probabilistic interpretation of
quantum mechanics. The new inner product also plays a crucial role in
our formulation of quantum statistical thermodynamics,
which recovers results for the usual Hermitian set up.

We mostly concentrate in the scenario of two conserved quantities. This captures the physics contained in the general case of $k$ conserved quantities.
Consider the thermal state
of a system with Hamiltonian $H$ and a conserved observable $K$ (i.e., $[K,H]=KH-HK=0$), described by the density matrix $\rho$.
Its energy and $K$ expectation values are given, respectively, by  $\langle H\rangle=\Tr H\rho$ and $
\langle K\rangle=\Tr K\rho.$ There are many thermal states, (also known as mixed states),
with this average energy and $K$ expectation value, and the {\it thermal equilibrium state}
is the one which maximizes the von Neumann entropy $S_\rho=-\Tr\rho\log\rho,$ subject to the considered expectation values of the energy and of $K$.

There are two ways to obtain the equilibrium thermal state: by maximizing the entropy,
and  in this case the inverse temperature
is the Lagrange multiplier which fixes the energy, or by minimizing the free energy,
and then the absolute temperature characterizes the associated heat source. More generally, in the first case,
the system is isolated and  the parameters playing the role of the Lagrange multipliers, which control the conserved quantities,
$H$ and $K$, are fixed in such a way as to preserve the values of these quantities. In the second case,
the system is not isolated and
the corresponding parameters characterize the sources of the conserved quantities,
which interact with the system.

The article is organized as follows.
In Section \ref{S2}, useful prerequisites are presented.
In Section \ref{S3}, our proposed formalism for the non-Hermitian context is given.
In Section \ref{S4}, the maximum entropy principle is investigated.
In Section \ref{S5}, our conclusions are summarized, and the difficulties of
an extension to the infinite dimensional context are sketched.
\section{Prerequisites}\label{S2}
The von Neumann formalism of quantum statistical  physics is established
 in the language of Hilbert spaces. 
{\it Quantum observables} are {\it self-adjoint}  (synonymously, {\it Hermitian}) operators acting on a Hilbert space $\cal H$.
We denote by $H_n$ the set f $n\times n$ Hermitian matrices.
A {\it density matrix} is a positive semidefinite matrix with unit trace.
Density matrices with rank 1 describe {\it pure states}, while those
with rank greater than 1 describe {\it mixed states} of the system.
Quantum probability measures are described by the eigenvalues of density matrices.

The statistical expectation value, or {\it average value}, of the observable $A$ for the state $\rho$ is given by
$$\langle A\rangle_\rho=\Tr(A\rho),$$
and the {\it von Neumann entropy} is equal to
$$S_\rho=-\Tr (\rho~\log \rho)=-\sum_k\eta_k\log\eta_k,$$
where the $\eta_k$ are the eigenvalues of $\rho$.
By convention, $0\log 0=0.$

{\it Gibbs states} 
describe the equilibrium states of
classical thermodynamics.
They maximize the entropy $S_\rho$ under the condition $E=\Tr(H\rho)$,
 and minimize $E$ for fixed entropy $S_\rho$.

Let $A_i,~i=1,\ldots,n$ be Hermissian matrices and assume that $\left\{\1,A_1,\ldots,A_n\right\}$ are linearly independent, where $\1$ is the identity matrix.
A {\it generalized thermal equilibrium state} is described by a density matrix of the form
$$\rho_\beta=\frac{\e^{-(\beta_1A_1+\ldots+\beta_nA_n)}}{\Tr{\e^{-(\beta_1A_1+\ldots+\beta_nA_n)}}},
\quad\beta_1, \ldots,\beta_n\in\R.$$
The Hermitian matrices $A_i$ represent {\it conserved quantities}, that is to say,
the Hamiltonian or observables that commute with it,
and the $\beta_i$ may be regarded as {\it generalized inverse temperatures} associated to the
conserved quantities. The function
$$Z:=\Tr{\e^{-(\beta_1A_1+\ldots+\beta_nA_n)}} $$
is the {\it generalized partition function} \cite{wich} and $\log Z$ is the {\it log partition function}.

Consider the Gibbs state $\rho_\beta$.
The exponential family ${\cal E}=\{\rho_\beta:\beta\in\R^n\}$ has two natural charts.
The first chart is the inverse map to
$$\alpha:\R^n\rightarrow,{\cal E},\quad\beta\rightarrow\rho_\beta,$$
so that the chart is $\alpha^{-1}:{\cal E}\rightarrow\R^n.$ The second chart is the restriction to $\cal E$
of the linear map
$$\E:H_n\rightarrow\R^n,\quad B\rightarrow\Tr(B A_i)_{i=1}^n,$$
i.e.,  the chart is $\E_{\cal E}\rightarrow\R^n.$ Recall that $\alpha^{-1}({\cal E})=\R^n$
and that $ \E({\cal E})={\rm int}(W(A_1,\ldots,A_n))$, the interior of {\it the joint numerical range} of the matrices $A_1,\ldots,A_n,$
which is defined as
$$W(A_1,\ldots,A_n)=\{\langle A_1,\ldots,A_n)\psi,\psi\rangle:\psi\in{\cal H},~\langle\psi,\psi\rangle=1\}.$$
It is well-known that ${\cal E}\circ\alpha:\R^n\rightarrow{\rm int}(W(A_1,\ldots,A_n))$ is an analytic diffeomorphism \cite{weis,wich}.

Let us consider the set
$$\Omega_{\beta_0}:=\left\{\Tr\left((A_1,\ldots,A_n)\frac{\e^{-\beta_1A_1-\ldots-\beta_nA_n}}{\Tr~\e^{-\beta_1A_1-\ldots-\beta_nA_n}}\right)
:\sqrt{\beta_1^2+\ldots+\beta_n^2}\leq\beta_0;~\beta_1,\ldots,\beta_n\in\R\right\},$$
Its boundary $\partial\Omega_{\beta_0}$ is an analytic hypersurface in $\R^n.$ For $\beta'_0<\beta_0$, we have
$$\Omega_{\beta'_0}\subset\Omega_{\beta_0}.$$
It should be noticed that $\Omega_{\beta_0}$ is convex.

\section{Non-Hermitian formalism}\label{S3}
Assume by now that $\cal H$ has finite dimension $n$.
The operator $H$, which is assumed to have real discrete spectrum,
 and its {\it adjoint} $H^\dag$ have the same eigenvalues.
We denote by $\psi_i$ the eigenvector of $H$ associated to the (non degenerate) eigenvalue $\lambda_i$,
and by  $\widetilde\psi_i$ the eigenvector of $H^\dag$ associated to the same eigenvalue $\lambda_i$.
The sets of eigenvectors $\left\{\psi_k\right\}$ and $\left\{\widetilde\psi_k\right\}$ are biorthogonal,
$\langle\psi_k,\widetilde\psi_l\rangle=0$ if $k\neq l$, and form bases
of $\cal H$, since this space has a  finite dimension $n$.
We orthonormalize the bases $\{\psi_i\},~\{\widetilde\psi_i\} ,$ so that
$$\langle\psi_k,\widetilde\psi_l\rangle=\delta_{kl},$$
where  $\delta_{kl}$ denotes  the {\it Kronecker} symbol (=1 for $k=l$ and 0 otherwise).
Let us define the
matrix $D=[D_{ij}]_{i,j=1}^n$ (we use synonymously the terms operator and matrix) such that
\begin{equation}\widetilde\psi_i=D\psi_i,~~i=1,\ldots,n.\label{D}\end{equation}
\begin{pro}
The matrix $D$ is positive definite. 
\end{pro}
\begin{proof}For
$\psi=\sum_{k=1}^n x_k\psi_k$,
we have
\begin{eqnarray*}&&\langle D\psi,\psi\rangle=\sum_{k,l=1}^n\langle D\psi_k,\psi_l\rangle x_k\overline x_l
=\sum_{k,l=1}^n\langle\widetilde\psi_k,\psi_l\rangle x_k\overline x_l
=\sum_{k,l=1}^n \delta_{kl} x_k\overline x_l=\sum_{k=1}^n | x_k|^2\geq0,
\end{eqnarray*}
and is zero if and only if $\psi=0.$
\end{proof}

{We define a new inner product}
$$\langle\phi,\psi\rangle_D:=\langle D\phi,\psi\rangle,
~~\text{for any}~~ \phi,~\psi\in{\cal H}.$$
For commodity, we say that this is the inner product with {\it metric} $D$,
or simply the $D$-{\it inner product}.
\begin{pro}\label{P22}
The non-Hermitian Hamiltonian $H$ is Hermitian relatively to the $D$-inner product.
\end{pro}
\begin{proof}
For any $\psi=\sum_{k=1}^nx_k\psi_k$ in $\cal H$ we have, 
$$\langle DH\psi,\psi\rangle\in\R,$$
because
\begin{eqnarray*}&&\langle DH\psi,\psi\rangle=\sum_{k,l=1}^n\langle DH\psi_k,\psi_l\rangle x_k\overline x_l
=\sum_{k,l=1}^n\lambda_k\langle\widetilde\psi_k,\psi_l\rangle x_k\overline x_l
\\&&
=\sum_{k,l=1}^n\lambda_k \delta_{kl} x_k\overline x_l=\sum_{k=1}^n\lambda_k | x_k|^2\in\R.
\end{eqnarray*}
Thus,
\begin{eqnarray*}
\langle H\psi,\psi\rangle_D=\langle DH\psi,\psi\rangle=\langle H^\dag D\psi,\psi\rangle=
\langle D\psi,H\psi\rangle=\langle\psi,H\psi\rangle_D.
\end{eqnarray*}
\end{proof}

\begin{pro}\label{T21}
The non-Hermitian Hamiltonian $H$ has real eigenvalues if and only if there exists a positive
definite matrix $D_0$ such that
$$D_0H=H^\dag D_0.$$
\end{pro}
\begin{proof}
Consider $D_0=D$ above defined. Observe that from Proposition \ref{P22} it follows that
$$\langle DH\psi,\psi\rangle=\langle D\psi,H\psi\rangle=\langle H^\dag D\psi,\psi\rangle,$$
for any $\psi\in{\cal H}$ so that $DH=H^\dag D.$

Suppose next that there exists $D_0$ positive definite and $D_0H=H^\dag D_0.$ Let $\psi_k$ be an eigenvector of $H$ associated with the eigenvalue
$\lambda_k$, claimed to be real, that is $H\psi_k=\lambda_k\psi_k$,
 so that
\begin{eqnarray}
&&D_0H\psi_k=H^\dag D_0\psi_k=\lambda_kD_0\psi_k.
\end{eqnarray}
Since $$\lambda_k
=\frac{\langle\psi_kD_0H,\psi_k\rangle}{\langle D_0\psi_k,\psi_k\rangle}
=\frac{\langle\psi_k H^\dag D_0,\psi_k\rangle}{\langle D_0\psi_k,\psi_k\rangle}
,$$
it follows that $\lambda_k$ is real.
\end{proof}
\begin{pro}
The Hamiltonian $H$ is similar to a Hermitian operator $H_0$, under the similarity
$D^{1/2}H{D^{-1/2}}.$
\end{pro}
\begin{proof}
The condition $DH=H^\dag D$ implies that
$$D^{1/2}HD^{-1/2}=D^{-1/2}H^\dag{D^{1/2}}=H_0.$$
Thus, $H_0$ is Hermitian and
$$H=D^{-1/2}H_0{D^{1/2}}.$$
\end{proof}


Suppose that the statistical properties of the
physical system we are concerned  are described~by a density matrix $\rho$, with $\Tr \rho=1,$
which is positive semidefinite
(notation, $\rho\geq0$)
under the metric $D$,
$$\langle\rho\psi,\psi\rangle_D=\langle D\rho\psi,\psi\rangle\geq0,~~\text{
for any}~~ \psi\in{\cal H},$$
{so that $D\rho=\rho^\dag D.$} The {\it statistical expectation value} of the energy is
$\langle H\rangle=\Tr( H\rho)$ and the {\it entropy}
is $S_\rho=-\Tr(\rho\log\rho).$
We define the {\it (non-equilibrium) free energy} of the system as in the standard case,
\begin{equation}F=\Tr(H\rho)+\Tr(\rho\log\rho).\label{F}\end{equation}

Notice that these definitions, used in the standard Hermitian set up, are still meaningful
in the present context. In fact, since
$DH=H^\dag D$
and $D\rho=\rho^\dag D$, we have, by the ciclicity of the trace,
\begin{eqnarray*}
\Tr(\rho H)=\Tr(D\rho HD^{-1})=\Tr(\rho^\dag DHD^{-1}=\Tr(\rho^\dag HDD^{-1})=\Tr(H^\dag\rho^\dag)\in\R.
\end{eqnarray*}
Similarly,
\begin{eqnarray*}
\Tr(\rho \log\rho)=\Tr(D\rho \log\rho D^{-1})=\Tr(\rho^\dag\log\rho^\dag DD^{-1})=\Tr(\rho^\dag\log\rho^\dag)\in\R.
\end{eqnarray*}

The free-energy is related to the partition function as follows
$$F=-T\log~Z.$$
{The {\it first law of thermodynamics} means that the energy is an {\it additive state function}
which is conserved, i.e., it remains constant.}
The energy expectation value is
$$\langle H\rangle
=-\frac{\d\log Z}{\d\beta}=\frac{\Tr H\e^{-\beta H}}{\Tr\e^{-\beta H}}
.$$

{The {\it second law of thermodynamics} means that the entropy is an {\it additive state function}
which increases when equilibrium is approached.}
The entropy is
$$S=\frac{\d}{\d T}(T \log Z)
=\log Z+\beta \frac{\Tr H\e^{-\beta H}}{\Tr\e^{-\beta H}}=-\Tr\left(\frac{\e^{-\beta H}}{\Tr\e^{-\beta H}}\log\frac{\e^{-\beta H}}{\Tr\e^{-\beta H}}\right),
$$
where $T=1/\beta.$

Consider next the existence of a conserved quantity $K$, that is, an observable with real eigenvalues, which,
be definition, commutes with $H$, $[H,K]=0.$ Thus, $H$ and $K$ have common eigenvectors, so that
they are both $D$-Hermitian, and
we have  $DH=H^\dag D$ and $DK=K^\dag D.$ If, moreover, $\beta,~\zeta\in\R,$ the equilibrium statistical expectation values
 of $H$ and $K$ may be defined as
$$\langle H\rangle=\frac{\Tr H\e^{-\beta H-\zeta K}}{\Tr\e^{-\beta H-\zeta K}}
,~~
\text{and}
~~\langle K\rangle=\frac{\Tr K\e^{-\beta H-\zeta K}}{\Tr\e^{-\beta H-\zeta K}}
,$$
because ${\Tr H\e^{-\beta H-\zeta K}},~{\Tr K\e^{-\beta H-\zeta K}},~{\Tr\e^{-\beta H-\zeta K}}\in\R.$
In fact, as $D>0,$ $DH=H^\dag D,$ $DK=K^\dag D,$ we have
$$\Tr(\beta H+\zeta K)=\Tr D (\beta H+\zeta K) D^{-1}=\Tr (\beta H+\zeta K)^\dag DD^{-1}=\Tr(\beta H+\zeta K)^\dag.$$

We easily find
\begin{eqnarray*}&&\Tr(\beta H+\zeta K)^k=\Tr (D(\beta H+\zeta K)^k D^{-1})=\Tr((\beta H^\dag+\zeta K^\dag) D (\beta H+\zeta K)^{(k-1)}D^{-1})\\
&&=\ldots
=\Tr(({\beta H^\dag+\zeta K^\dag})^kDD^{-1})=\Tr{(\beta H^\dag+\zeta H^\dag)}^k.\end{eqnarray*}
Since
$$\e^{-\beta H-\zeta K}=\sum_{k=0}^\infty\frac{(-\beta~H-\zeta K)^k}{k!},$$
the claim follows.

According to the maximum entropy (MaxEnt) principle, the equilibrium thermal state of an isolated system
is  determined by maximizing the entropy of the system subject to constrained values of the conserved quantities
$H$ and $K$.
The Lagrange multipliers which control the conserved quantities
are fixed in such a way as to preserve their values.

%
%
\begin{pro}\label{T32}
If $H$ and $K$ are $D$-Hermitian, $\beta,~\zeta\in\R,$ and $[H,K]=0,$ then
$$\langle H\rangle
=\frac{-\partial \log Z}{\partial\beta},$$
and
$$\langle K\rangle
=\frac{-\partial \log Z}{\partial\zeta}.$$

\end{pro}
\begin{proof}
Since $[H,K]=0$, we may write
 $$\frac{\Tr H\e^{-\beta H-\zeta K}}{\Tr\e^{-\beta H-\zeta K}}
=\frac{-\partial \log Z}{\partial\beta},$$
and
$$\frac{\Tr K\e^{-\beta H-\zeta K}}{\Tr\e^{-\beta H-\zeta K}}
=\frac{-\partial \log Z}{\partial\zeta}.$$
The result follows.
\end{proof}


The following question naturally arises. Is it legitimate to describe an isolated
system by a Gibbs state? We try to provide a partial answer to it.
Let us replace the Hilbert space $\cal H$ by
$${\cal H}_{comp}={\cal H}\otimes{\cal H}\otimes\ldots\otimes{\cal H},$$
where the number of factors in the tensorial product is $N$,
and let us consider the {\it composed system} which is constituted by $N$ partial systems and is described by the Hamiltonian
$$H_{comp}=H\oplus H\oplus\ldots\oplus H,$$
where the number of summands is $N$. For simplicity, we simply denote by $H$, according to its position in the direct sum,
each one of the operators $(H\otimes I\otimes\ldots\otimes I),~~ (I\otimes H\otimes I\otimes\ldots\otimes I),~~\ldots,~~
(I\otimes I\otimes\ldots\otimes I\otimes H)$ acting on ${\cal H}_{comp}$. It  is clear that the energy expectation value,
free energy, entropy and energy variance of the composed system are $N$ times the corresponding quantities relative to the
partial system. Since the statistical error is determined by the square root of the variance, it is clear that, if $N$ is large enough,
the Gibbs state safely describes the isolated composed system.
\section{MaxEnt principle}\label{S4}
\def\be{{\beta}}
The following minimum free energy (or maximum entropy) inequality holds.
\begin{teor}\label{T1}
Let the Hamiltonian $H$ be non-Hermitian with real simple eigenvalues. Assume
$D$ as defined in (\ref{D}) and  $\beta,\zeta\in\R$. If the density matrix $\rho$
and the operator $K$ are Hermitian relatively to the $D$-inner product,
 then
\begin{equation}\label{free_energy}-\log\Tr\e^{-\beta H-\zeta
K}\leq \Tr\rho(\beta H+\zeta K+\log\rho)
,\end{equation} with equality occurring if and only if
\begin{equation}\label{rho0}\rho=\rho_0:={\e^{-\beta H-\zeta K}\over\Tr\e^{-\beta H-\zeta
K}}.
\end{equation}
\end{teor}

\begin{proof} According to the hypothesis, $\rho$ and $D$ satisfy
$D\rho=\rho^\dag D,~~ DK=K^\dag D$.
Let the matrix $U$ satisfy
$$\langle DU\psi,U\psi\rangle=\langle D\psi,\psi\rangle,$$
for all $\psi\in{\cal H}.$
That is, the relation
$$DU=(U^\dag)^{-1}D,$$
holds and implies that $D^{1/2}UD^{-1/2}$ is unitary.
Moreover, we may write
$$U=\e^{iT},$$
where $T$ is such that $DT=T^\dag D.$ Thus,
the matrix $D^{1/2}TD^{-1/2}$ is Hermitian.
Recall that the group ${\cal U}_n$ of unitary matrices is compact, so that the set
$D^{-1/2}{\cal U}_nD^{1/2}$, to which $U$ belongs, is also compact.
Let us replace $\rho$ by $U\rho U^\dag$.
Obviously, $\Tr\rho\log\rho$ remains unchanged. The
minimum of
\begin{equation}\label{FF}\Tr U\rho U^\dag(\beta H+\gamma K+\log(U\rho
U^\dag))\end{equation} with respect to $U$, occurs when
$$[U\rho U^\dag,(\beta H+\gamma K)]=0,$$
where, as usual, $[X,Y]=XY-YX$ denotes the commutator of $X$ and
$Y$. This easily follows, assuming that the maximum is reached
when $U$ is replaced by
$\exp(i\epsilon T)U$, where $T$ is an arbitrary Hermitian matrix,
$\epsilon$ is a sufficiently small real number, and
(\ref{FF}) is expanded up to first order in $\epsilon$. Since this term must
vanish for any $T$, we conclude that $[U\rho U^*,(\beta H+\gamma
K)]=0$. Therefore, the  matrices $D^{1/2}U\rho U^\dag D^{-1/2}$ and $D^{1/2}(\beta
H+\gamma K)D^{-1/2}$ are simultaneously unitarily diagonalizable. Let us
denote the real eigenvalues of $\rho$ and $(\beta H+\zeta K)$,
respectively, by $\eta_1,\ldots,\eta_n$ and by
$\lambda_1,\ldots,\lambda_n,$ so that we may write
\begin{eqnarray*}
&&\Tr U\rho U^*(\beta H+\gamma K+\log(U\rho
U^*))=\sum_j(\eta_j\lambda_j+\eta_j\log\eta_j)\\
&&=\sum_j\left(\eta_j\lambda_j+\eta_j\log\eta_j+\log\sum_k\e^{-\lambda_k}\right)-\sum_j\eta_j\log\sum_k\e^{-\lambda_k}\\
&&
=\sum_j\eta_j\left(\log\left(\eta_j\e^{\lambda_j}
\sum_k\e^{-\lambda_k}\right)-\log\sum_k\e^{-\lambda_k}\right)
\\&&=\sum_j{\e^{-\lambda_j}\over\sum_k\e^{-\lambda_k}}\left(\eta_j\e^{\lambda_j}
\sum_k\e^{-\lambda_k}\right)\log\left(\eta_j\e^{\lambda_j}\sum_k\e^{-\lambda_k}\right)-\log\sum_j\e^{-\lambda_j}
\end{eqnarray*}
\begin{eqnarray*}
&&\geq\sum_j{\e^{-\lambda_j}\over\sum_k\e^{-\lambda_k}}\left(\eta_j\e^{\lambda_j}
\sum_k\e^{-\lambda_k}-1\right) -\log\sum_j\e^{-\lambda_j}\\&&
=-\log\sum_j\e^{-\lambda_j}=-\log\Tr\e^{-\beta H-\zeta K},
\end{eqnarray*}where the inequality follows because $x\log~x\geq
x-1$. Thus, we get the inequality in (\ref{free_energy}). It is
obvious that the equality occurs if and only if
$\eta_j=\e^{-\lambda_j}/\sum_k\e^{-\lambda_k}.$
\end{proof}

The previous theorem is valid even when we do not have separately $DH=H^\dag D,~DK=K^\dag D.$ It is enough
that $D(\beta H+\zeta K)=(\beta H^\dag+\zeta K^\dag)D.$ This is ensured by the reality of the eigenvalues of
$(\beta H+\zeta K)$.


\def\ga{{\gamma}}
\def\be{{\beta}}
\section{Determining the Gibbs state}
\begin{pro} The function $\log Z:\R^2\rightarrow\R$ such that $(\beta,\zeta)\rightarrow\log Z$, is convex.\label{5.1}
\end{pro}

\begin{proof}
We compute the Hessian of $\log Z(\beta,\zeta).$ We find
$$Hess=\left[\begin{matrix}{\rm Co}_{H,H}&{\rm Co}_{H,K}\\
{\rm Co}_{H,K}&{\rm Co}_{K,K}\end{matrix}\right],$$
where
$${\rm Co}_{H,H}=
{\partial^2\log ~Z(\beta,\zeta)\over\partial\beta^2},$$
$${\rm Co}_{K,K}=
{\partial^2\log ~Z(\beta,\zeta)\over\partial\zeta^2},$$
$${\rm Co}_{H,K}=
{\partial^2\log ~Z(\beta,\zeta)\over\partial\beta\partial\zeta}.$$
Now,
$$
{\partial^2\log ~Z(\beta,\zeta)\over\partial\beta^2}=\langle H^2\rangle-\langle H\rangle^2,$$
$$
{\partial^2\log ~Z(\beta,\zeta)\over\partial\zeta^2}=\langle K^2\rangle-\langle K\rangle^2,$$
$$
{\partial^2\log ~Z(\beta,\zeta)\over\partial\beta\partial\zeta}=\langle HK\rangle-\langle H\rangle\langle K\rangle.$$
Thus, the Hessian coincides with the covariance matrix, which is positive definite and the result follows.
 \end{proof}

\begin{pro}\label{T43} Under the hypothesis of Theorem \ref{T1} and $[H,K]=0,$ the function $F:\R^2\rightarrow\R^2$ such that
$$F(\beta,\zeta):=-\left({\partial\log Z\over\partial\beta},{\partial\log Z\over\partial\zeta}\right)$$
is injective.
\end{pro}
\begin{proof}
According to Proposition \ref{5.1}, the function $\log Z(\beta,\zeta)$ is convex,
implying that the function $F(\beta,\zeta)$
is injective. \end{proof}

\begin{pro} The function $S_{eq}:\R^2\rightarrow\R$ such that  $S_{eq}$ is
 the maximum entropy compatible with the expectation values $\langle H\rangle$,  $\langle K\rangle$
of the conserved quantities $H,~K,$
is concave.
\end{pro}
\begin{proof} Observe that
$$S_{eq}
=\log Z+\beta \langle H\rangle+\zeta  \langle K\rangle
,$$
is the Legendre transform of $\log(Z(\beta,\zeta)),$
which is convex by Proposition \ref{5.1}. Here $(\beta,\zeta)$ is the pre-image of $(\langle H\rangle,\langle K\rangle)$
under the function $F$ of Proposition \ref{T43}. The result follows.
\end{proof}

The maximum entropy inference problem deals with the  determination of
$(\beta,\zeta)$, from the knowledge of
$$x(\be,\zeta)=-{\partial\log ~Z(\beta,\zeta)\over\partial\beta}~~\text{ and}~~
y(\be,\zeta)=-{\partial\log ~Z(\beta,\zeta)\over\partial\zeta}.$$
That is, searching the pre-image $(\be,\zeta)$ of the
function $F:\R^2\rightarrow\R^2$, such that
$$(\beta,\zeta)\rightarrow(x(\be,\zeta),y(\be,\zeta)),$$
is required.
According to Proposition \ref{T43}, $F$  is injective,
allowing the restoration of the maximizing matrix in
(\ref{rho0}). 
The parameters $\be,\zeta$ and the
constraints on $\langle H\rangle$ and  $\langle K\rangle$
are related according to
 $\Tr H\rho_0=\langle H\rangle$ and  $\Tr K\rho_0=\langle K\rangle$ for $\rho_0$ in (\ref{rho0}).
We may observe that the set of points $(x(\be,\ga),y(\be,\zeta))$,
associated with the Lagrange multipliers $\be,\zeta$,  for
$\be=\be_0\cos\theta,~\zeta=\be_0\sin\theta,~0\leq\theta\leq 2\pi,$
and fixed $\beta_0$, tends to the boundary of the {\it numerical range}
$$ W(H+iK)=\{\langle(H+iK)\psi,\psi\rangle:\psi\in{\cal H},\langle\psi,\psi\rangle=1\},$$
 when
$\beta_0\rightarrow+\infty$ \cite{bebianoFE},
Let us consider the families of curves
\begin{eqnarray*}
\Gamma_{\beta_0}=\{(x(\beta_0,\theta),(y(\beta_0,\theta)):-\pi\leq\theta<\pi\},~~~
\Gamma_{\theta}=\{(x(\beta_0,\theta),(y(\beta_0,\theta)):0\leq\beta_0<\infty\},
\end{eqnarray*}
where
$$x(\beta_0,\theta)=x(\beta,\zeta)|_{\beta=\beta_0\cos\theta,~\zeta=\beta_0\sin\theta} ~~\text{ and}~~
y(\beta_0,\theta)=y(\beta,\zeta)|_{\beta=\beta_0\cos\theta,~~\zeta=\beta_0\sin\theta}.$$
The maximum entropy inference problem is solved by determining $\beta_0,\theta$ from the intersection
$$\Gamma_{\beta_0}\cap\Gamma_\theta=(x(\beta_0,\theta),(y(\beta_0,\theta)).$$

The following schematic Example illustrates the choice of the specific Gibbs state which is determined
by the given expectation values of $H$ and $K$. The described procedure may be numerically implemented.
\begin{figure}[ht]
\centering
\includegraphics[width=.7\textwidth, height=0.7\textwidth]
{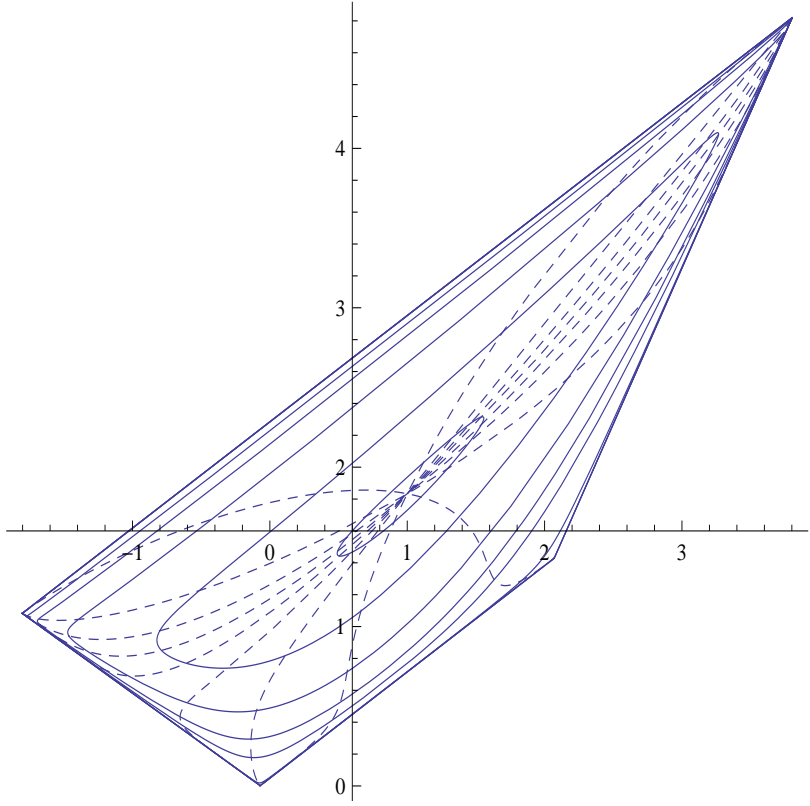} \caption{The curves
$x(\be,\zeta),y(\be,\zeta)$ for fixed values of  $\beta_0$ and variable values of $\theta$ (full lines);
and for fixed values of $\theta$ and variable values of $\beta_0$ (dashed lines).
The horizontal and the vertical axes, represent,
respectively, $x(\be,\zeta)$ and $y(\be,\zeta)$.} \label{fig00}
\end{figure}
\begin{Exa}\label{Exa}
Let us consider the observables
$$H=\left[\begin{matrix}1&1&0&0\\
3&1&1&0\\
0&3&1&1\\
0&0&3&1\end{matrix}\right],\quad K=\frac{1}{3}H^2.$$
and
$
A=H+iK.$
It may be easily seen that the numerical range of $A$ is a quadrilateral.

Let
$$
x(\beta_0,\theta)+iy(\beta_0,\theta):=\Tr\left(\frac{\e^{-\beta_0(\cos\theta H+\sin\theta K)}}{\Tr {\e^{-\beta_0(\cos\theta H+\sin\theta K)}}}(H+i K)\right),\quad x(\beta_0,\theta),~y(\beta_0,\theta)\in\R.
$$

Fixing $\beta_0$ and varying $\theta$ we obtain a closed curve
surrounding the point of maximal entropy, $1+11i/6$. Fixing $\theta$ and varying $\be_0$, we
obtain  curves connecting the point $1+11i/6$ with corners of $W(H+iK)$.
The full curves displayed in Figure 1 are
for 
$\beta_0=0.1,0.5,1,1.5,2,4,8,16,32$ and $0<\theta<2\pi$. For $\beta_0=8,16,32$  the lines are not distinguishable
and coincide with the boundary of $W(H+iK).$
The displayed dashed curves are  for $\theta=\pi/8,\pi/4,3\pi/8,\pi/2,5\pi/8,3\pi/4,7\pi/8,\pi$ 
and $-32<\beta_0<32$.
In the limit
$\beta_0\rightarrow+\infty,$ the boundary of $ W(H+iK))$ is obtained
i.e., the limit of the solution $\rho_0$ corresponds for almost any
$\theta$ to a pure state, with entropy $S=0$.
\end{Exa}


%

\begin{figure}[ht]
\centering
\includegraphics[width=.7\textwidth, height=0.7\textwidth]
{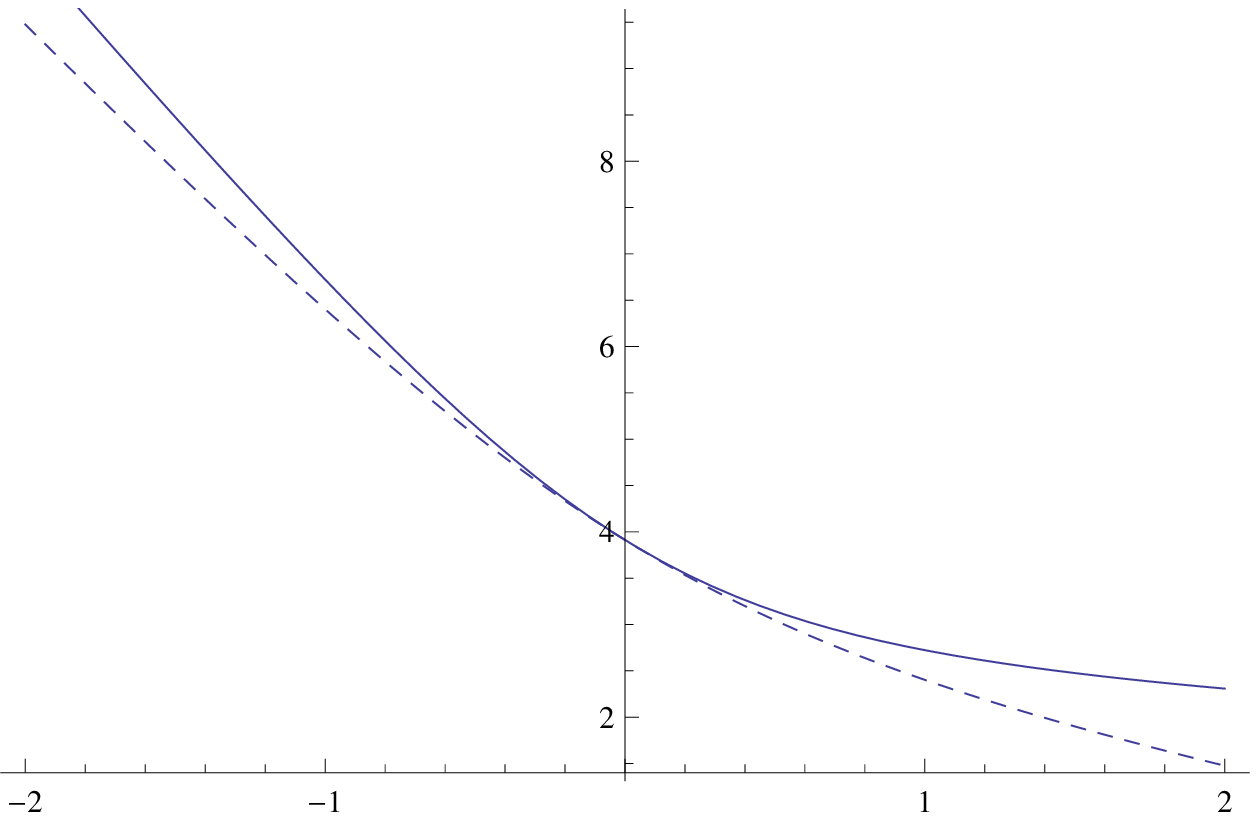} \caption{ {Illustrating the concavity of the $\log Z$ vesus $\beta$.}
Results obtained for $n=50.$
 Hermitian case, $d=0$, full lines, and non-Hermitian case,  $d=\sqrt7/4$, dashed lines. }  \label{fig11}
\end{figure}
\begin{figure}[ht]
\centering
\includegraphics[width=.7\textwidth, height=0.7\textwidth]
{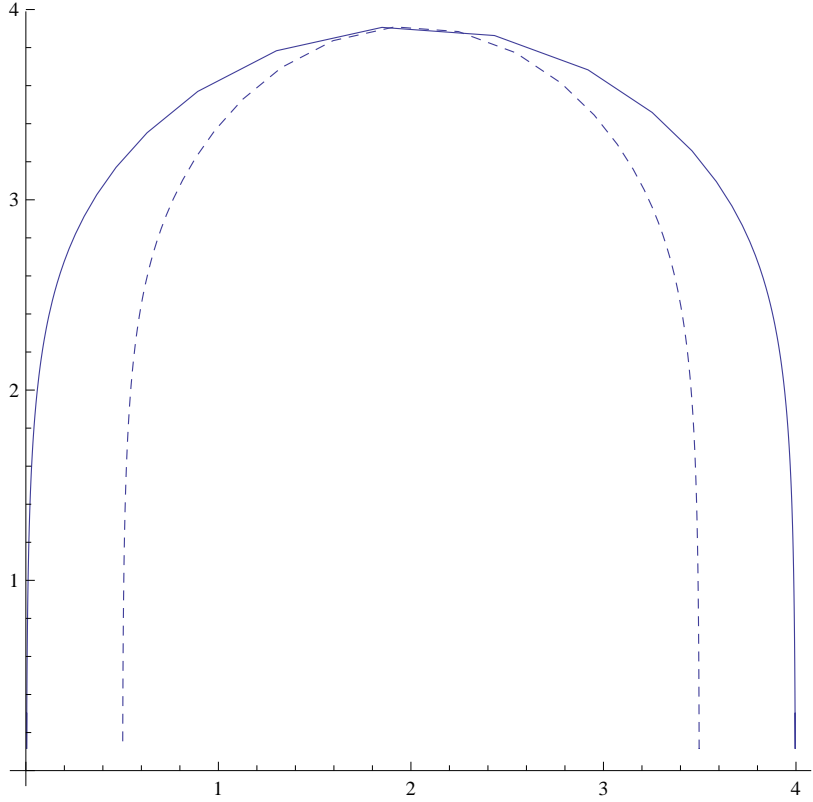} \caption{ Illustrating the concavity of the maximum entropy vs. $\langle H\rangle$.
Results obtained for $n=50.$
 Hermitian case, $d=0$, full lines, and non-Hermitian case,  $d=\sqrt7/4$, dashed lines. }  \label{fig22}
\end{figure}

\begin{Exa}
We consider next a model whose Hamiltonian is a
Toeplitz matrix $K_n$,  which is
non-Hermitian for $d\neq0$. To ensure the reality of the spectrum we impose the condition $|b|<1. $
\begin{equation}K_n=\left[\begin{matrix}2&1-d&0&0&\ldots&0\\
1+d&2&1-d&0&\ldots&0\\
0&1+d&2&1-d&\ldots&0\\
0&0&1+d&2&\ldots&0\\
\vdots&\vdots&\vdots&\vdots&\ddots&\vdots\\
0&0&0&0&\ldots&2
\end{matrix}\right],~~d\in\R,~|d|<1.\label{nonHermitian}\end{equation}
Its eigenvalues are $$\lambda_k=2-2\sqrt{1-d^2}\cos\frac{k\pi}{n+1}.$$
There exists $D\geq0$ such that $DK_n=K_n^\dag D.$
Figure 2 illustrates the convexity of $\log Z$ vs. $\beta$. Figure 3 illustrates  the concavity of the maximum entropy vs. $\langle H\rangle$.
In order to approach the partition function we use te following result on the Euler-McLaurin expansion.
\end{Exa}
\begin{pro}\label{L6}
Let $n$ be a positive integer and let $f$ be a real function
defined in the real interval $[0,1]$,
being of class $C^\infty$ in
$[0,1),$  $f'(1)$ exists but $f''(1)$ does not exist. Then%
\begin{equation}\label{lemma}
\sum_{k=1}^{n}f\left({k\over n}\right)=n\int_{0}^{1}f(x)\d x
+\frac{1}{2}(f(1)-f(0))+\frac{1}{12n}(f^{\prime }(1)-f^{\prime }(0))
+R_{n},
\end{equation}
with
\begin{eqnarray}\label{Rn}R_n=-{1\over2n}\int_{1/n}^1
B_2(\{nx\})f''(x)\d x.\end{eqnarray}
where $B_2(x)=x^2-x+1/6$ is the second Bernoulli polynomial and $\{t\}$ denotes the fractional part of $t$.
\end{pro}

It is known that
$$\int_0^{n+1}\e^{f+h\cos(k\pi/(n+1))}\d k=\e^f(1+n)~_0I(h)$$
where $_nI(z)$ is the modified Bessel function of first kind.

By the Euler-MacLaurin formula, we obtain
\begin{eqnarray*}
&&\sum_{k=1}^n\e^{-\beta\lambda_k}\\
&&=-\e^{-\beta(b-\sqrt{ac})}+\sum_{k=1}^{n+1}\e^{-\beta\lambda_k}\\
&&\approx \e^{-\beta b} (1 + n) ~{_0I(\beta\sqrt{ac})}-\frac{1}{2}\left(\e^{-\beta(b+\sqrt{ac})}+\e^{-\beta(b-\sqrt{ac})}\right).
\end{eqnarray*}



\section{Concluding remarks}\label{S5}
If $H$ lives in an infinite dimensional Hilbert space, different situations may occur,
such as the metric operator or its inverse, or both, being, possibly, unbounded. The eigenstates of the Hamiltonian $H$
and of $H^\dag$ are biorthogonal but they cannot form bases of $\cal H$ \cite{bagarello*}. The existence of a
bounded operator with bounded inverse mapping some orthonormal bases of $\cal H$ into the sets
$\{\psi_k\}$ and $\{\widetilde\psi_k\}$ is not guaranteed, {\it a priori}. Thus, the previous procedure should be reconsidered
carefully. We notice, however, that from the point of view of physics,
the full Hilbert space $\cal H$ may not be needed.
Nothing guarantees that all vectors in $\cal H$ have physical meaning.
Let
${\cal S}:={\rm span}\{\psi_k\}$,
${\widetilde{\cal S}}:={\rm span}\{\widetilde\psi_k\}.$
Only vectors 
$\psi\in{\cal S}$ represent physical states.
Although $D$ is not defined in $\cal H$, it goes from $\cal S$ to $\widetilde{\cal S}.$
The operators $H,~K,~\rho$ go from $\cal S$ to $\cal S$,
the operators $H^\dag,~K^\dag,~\rho^\dag$ go from $\widetilde{\cal S}$ to $\widetilde{\cal S}$.
The {\it Physical Hilbert space} is the set $\cal S$ endowed with the inner product $\langle D\cdot,\cdot\rangle$ \cite{mostafa}.

Summarizing,
for $n$ finite, ${\cal H}={\cal S}=\widetilde{\cal S}.$
For $n=\infty,$ this is not so, but the proposed definitions are still meaningful.

\end{document}